\newif\ifextended
\extendedtrue

\ifextended
  \documentclass[conference,onecolumn]{IEEEtran}
\else
  \documentclass[conference]{IEEEtran}
\fi

\usepackage{amsmath,amsfonts,amssymb,amsthm}
\usepackage{graphicx}
\usepackage{color,colortbl}

\makeatletter
\let\MYcaption\@makecaption
\makeatother

\usepackage[font=footnotesize]{subcaption}

\makeatletter
\let\@makecaption\MYcaption
\makeatother

\graphicspath{ {./figures/} }

\newtheorem{theorem}{Theorem}
\newtheorem{lemma}{Lemma}

\newcommand{\remove}[1]{}
\renewcommand{\epsilon}{\varepsilon}

\definecolor{user1}{rgb}{0.8,1,1}
\definecolor{user2}{rgb}{1,0.8,1}

\ifextended
  \newcommand{\appref}[1]{Appendix~\ref{#1}}
  \newcommand{\appreffirst}{the Appendices}
\else
  \newcommand{\appref}[1]{\cite{extended}}
  \newcommand{\appreffirst}{the extended version of this paper \cite{extended}}
\fi

\IEEEoverridecommandlockouts

\begin{document}

\title{A Layered Caching Architecture\\for the Interference Channel}

\author{%
\IEEEauthorblockN{Jad Hachem}
\IEEEauthorblockA{University of California, Los Angeles\\
Email: jadhachem@ucla.edu}
\and
\IEEEauthorblockN{Urs Niesen}
\IEEEauthorblockA{Qualcomm NJ Research Center\\
Email: urs.niesen@ieee.org}
\and
\IEEEauthorblockN{Suhas Diggavi}
\IEEEauthorblockA{University of California, Los Angeles\\
Email: suhas@ee.ucla.edu}
\ifextended
\thanks{A shorter version of this paper is to appear in IEEE ISIT 2016.}
\fi
\thanks{This work was supported in part by NSF grant \#1423271.}
}

\maketitle

\begin{abstract}
Recent work has studied the benefits of caching in the interference channel, particularly by placing caches at the transmitters.
In this paper, we study the two-user Gaussian interference channel in which caches are placed at both the transmitters and the receivers.
We propose a separation strategy that divides the physical and network layers.
While a natural separation approach might be to abstract the physical layer into several \emph{independent} bit pipes at the network layer, we argue that this is inefficient.
Instead, the separation approach we propose exposes \emph{interacting} bit pipes at the network layer, so that the receivers observe related (yet not identical) quantities.
We find the optimal strategy within this layered architecture, and we compute the degrees-of-freedom it achieves.
Finally, we show that separation is optimal in regimes where the receiver caches are large.

\end{abstract}

\section{Introduction}
\label{sec:intro}







Traditional communication networks are connection centric, i.e., they
establish a reliable connection between two fixed network nodes.
However, instead of a connection to a specific destination node, modern
network applications often require a connection to a specific piece of
content. Consequently, network architectures are shifting from being
connection centric to being content centric. These content-centric
architectures make heavy use of in-network caching and, in order to do so,
redesign the protocol stack from the network layer upwards~\cite{jacobson2012}.

Recent work in the information theory literature indicates that the
availability of in-network caching can also benefit the physical layer.
This information-theoretic approach to caching was introduced in the
context of the noiseless broadcast channel in \cite{maddah-ali2012},
where it was shown that significant performance gains can be obtained
using cache memories at the \emph{receivers}. The setting was extended
to the interference channel in \cite{maddah-ali2015interference}, which
presented an achievable scheme showing performance gains using cache
memories at the \emph{transmitters}. The achievable scheme from
\cite{maddah-ali2015interference} uses the cache memories to create many
virtual transmitters and improves transmission rate by performing
elaborate interference alignment between those virtual transmitters.

\begin{figure}
\centering
\includegraphics[width=.45\textwidth]{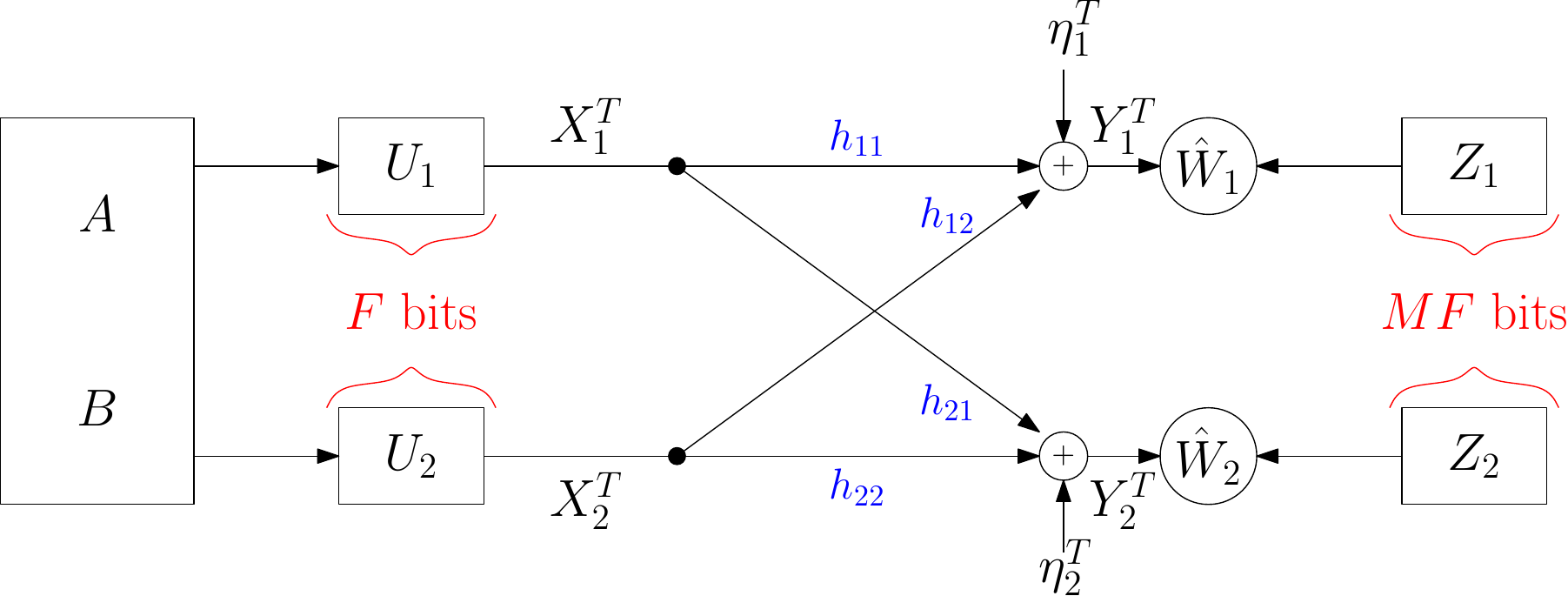}
\caption{The caching problem over the interference channel.
The server holds files $A$ and $B$, of size $F$ bits each, and caches parts of them in four memories $U_1$, $U_2$, $Z_1$, and $Z_2$.
The two users (circles) request files $W_1,W_2\in\{A,B\}$, and aim to recover them using the output of the interference channel and their respective caches.}
\label{fig:setup-gaussian}
\end{figure}

\begin{figure}
\centering
\begin{subfigure}{.45\textwidth}
\centering
\includegraphics[width=.8\textwidth]{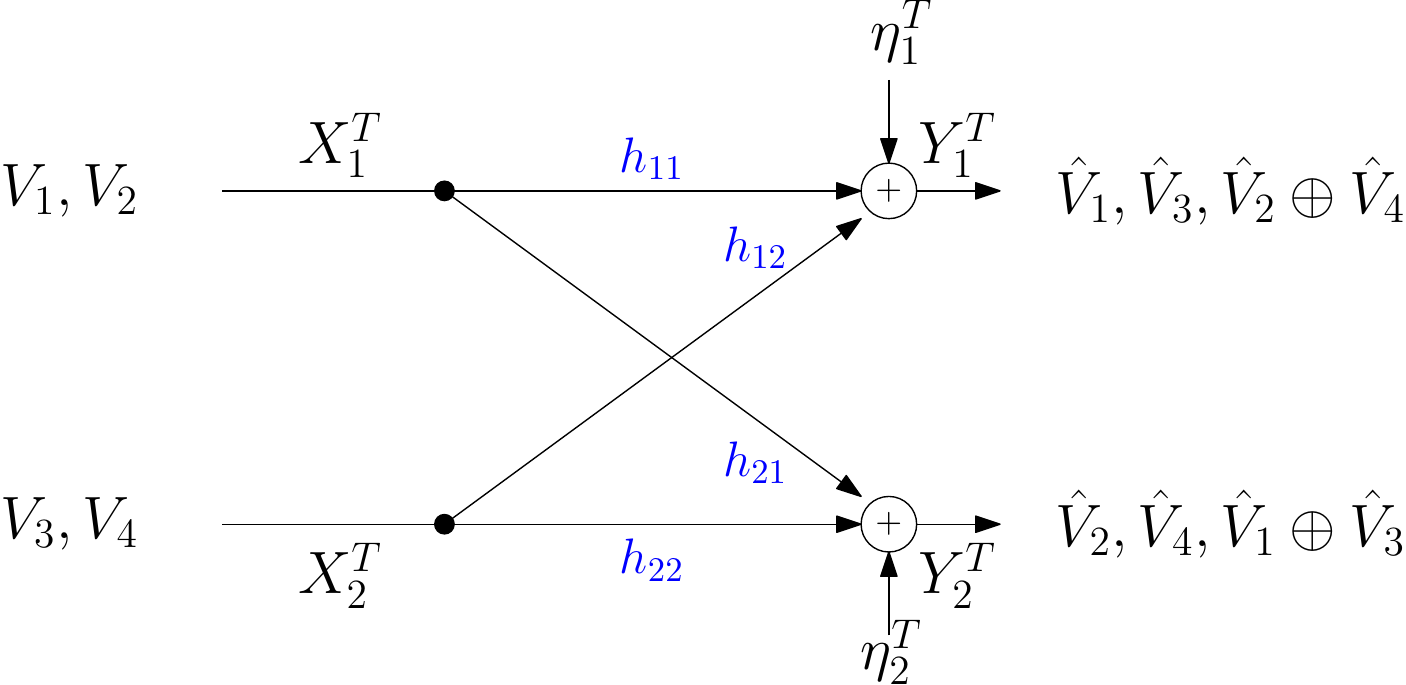}
\caption{Physical-layer view.}
\label{fig:physical-layer}
\end{subfigure}

\begin{subfigure}{.45\textwidth}
\centering
\includegraphics[width=\textwidth]{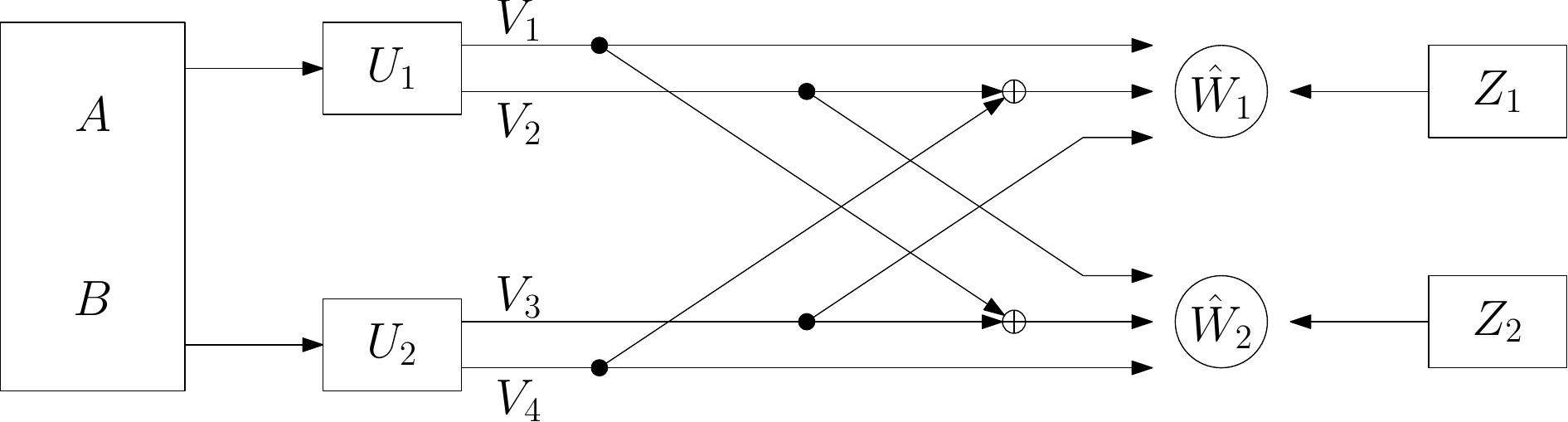}
\caption{Network-layer view.}
\label{fig:setup-deterministic}
\end{subfigure}

\caption{Physical and network layers of the system in \figurename~\ref{fig:setup-gaussian} under the proposed separation architecture.}
\label{fig:separation}
\end{figure}

In this paper we continue the study of cache-aided interference
channels, but we allow for caches at both the \emph{transmitters and
receivers} as shown in \figurename~\ref{fig:setup-gaussian}.
Furthermore, we propose a simpler, layered communication architecture,
separating the problem into a physical layer and a network layer as shown
in \figurename~\ref{fig:separation}. In
other words, we propose a redesign of the protocol stack from the
network layer downwards.

There are two seemingly natural network-layer abstractions for this
problem. The first treats the physical layer as a standard interference
channel and transforms it into two noninteracting error-free bit pipes.
The second treats the physical layer as an X-channel and transforms it
into four noninteracting error-free bit pipes. We argue that both of
these abstractions are deficient. Instead a more appropriate
abstraction needs to expose some of the algebraic structure of the
underlying physical layer to the network layer. More precisely, we
propose a network-layer abstraction consisting of four
\emph{interacting} error-free bit pipes as illustrated in
\figurename~\ref{fig:setup-deterministic}.

We derive optimal communication schemes for this layered architecture.
An interesting feature of these schemes is that they require coding
during both the content placement and delivery phases (whereas the
caching schemes studied in the prior literature utilize coding for only
one or the other). For the regime of large cache sizes, it turns out that
the layered architecture itself is fundamental, i.e., that the
separation of the communication problem into the two proposed layers is
without loss of optimality.

\paragraph*{Related work}

The information-theoretic framework for coded caching was introduced in
\cite{maddah-ali2012} in the context of the deterministic broadcast
channel.  This has been extended to online caching systems \cite{PMN13},
heterogeneous cache sizes \cite{wang2015fundamental},
unequal file sizes \cite{zhang2015filesize}, and improved converse
arguments \cite{ghasemi2015improved, sengupta2015improved}. Content
caching and delivery in device-to-device networks, multi-server
topologies, and heterogeneous wireless networks have been studied in
\cite{ji2013wireless,shariatpanahi2015multi,FemtoCaching,MLcodedcaching,HKDinfocom15}.
This framework has also been extended to hierarchical (tree) topologies in
\cite{Hcodedcaching}. More recently, it has been extended to
interference channels in \cite{maddah-ali2015interference}, where only
transmit caches were considered and several interesting schemes were
developed.

The paper is organized as follows. Section~\ref{sec:setup} provides the
formal problem setting and introduces the proposed layered communication
architecture.  Section~\ref{sec:results} presents a complete performance
characterization for this architecture.  A detailed description of the
network-layer processing together with optimality proofs are given in
\appreffirst.

\section{Problem Setting and Layered Architecture}
\label{sec:setup}
We study a system in which a server delivers files to two users across a Gaussian interference channel with the help of caches at all network nodes, described in Section~\ref{sec:setup-caching}.
We propose a layered communication architecture consisting of a physical and a network layer.
We introduce the physical layer in Section~\ref{sec:setup-physical} and the network layer in Section~\ref{sec:setup-network}.

In this paper, we restrict the number of files to just two.
We also fix the size of the \emph{transmitter} caches to the smallest size required for normal operation of the system.
This allows us to study the transmission rate of the files as a function of the receiver cache memory, without worrying about additional complexities arising from larger transmitter caches and a greater number of files.
In fact, the results turn out to be rather complex even in this simplified setting.
Extensions to this setup are a work in progress.

\subsection{The caching problem}
\label{sec:setup-caching}

Consider the setup in \figurename~\ref{fig:setup-gaussian}.
A server has two files $A$ and $B$ of size $F$ bits each.
The server is connected to two transmitters.
These in turn are connected to two receivers, referred to as users, through a Gaussian interference channel.

Communication occurs in two phases.
In the \emph{placement phase}, the server pushes file information to four caches denoted by $U_1$, $U_2$, $Z_1$, and $Z_2$.
Caches $U_1$ and $U_2$ are at transmitters 1 and 2, respectively, and can each store up to $F$ bits.
Caches $Z_1$ and $Z_2$ are at receivers 1 and 2, respectively, and can each store up to $MF$ bits.
The parameter $M\ge0$ is called the (normalized) \emph{memory size}.

In the subsequent \emph{delivery phase} each user requests one of the files (possibly the same).
Formally, users~1 and~2 request files $W_1,W_2\in\{A,B\}$, respectively, from the server.
Each transmitter $i$ then sends a length-$T$ sequence $X_i^T$ through the interference channel.
This message $X_i^T$ can depend only on the transmitter's cache content $U_i$.
In other words, the server itself does not participate in the delivery phase.

We impose a power constraint
\[
\frac1T \sum_{t=1}^T X_{i,t}^2 \le P
\]
on the transmitted sequence $X_i^T$.
Each receiver $i$ observes the channel output
\[
Y_i^T = h_{i1}X_1^t + h_{i2}X_2^T + \eta_i^T
\]
of the interference channel, where $\eta_{i,t}\sim\mathcal{N}(0,1)$ is iid additive Gaussian noise.
The receiver combines the channel output $Y_i^T$ with the cache content $Z_i$ to decode the requested file $\hat W_i$.

We define the \emph{rate} of the system as $R=F/T$.
Our goal is to characterize the trade-off between the rate $R$ and the receiver cache memory $M$ under the power constraint $P$.
Formally, we say that a tuple $(R,M,P)$ is achievable if there exists a strategy with rate $R$, receiver cache memory $M$, and power constraint $P$ such that
\[
\Pr\left\{ (\hat W_1,\hat W_2)\not=(W_1,W_2) \right\}
\to 0
\text{ as $T\to\infty$}
\]
for all possible demands $(W_1,W_2)\in\{A,B\}^2$.
We define the optimal rate function as
\[
R^*(M,P) = \sup\left\{ R : (R,M,P) \text{ is achievable} \right\}.
\]
We are particularly interested in the high-SNR regime and focus on the optimal degrees of freedom (DoF)
\[
d^*(M) = \lim_{P\to\infty} \frac{R^*(M,P)}{\frac12\log P}
\]
of the system.
It will be convenient to work with the inverse-DoF $1/d^*(M)$, since it is a convex function of $M$ \cite[Lemma~1]{maddah-ali2015interference}.

\subsection{Physical-layer view}
\label{sec:setup-physical}

We next describe the physical-layer view of the caching problem.
There are several possible strategies of how to perform the layer separation, each leading to a different physical-layer view.
We start by describing the advantages and disadvantages of some of them.

A natural strategy might be the complete separation of the physical and network layers, abstracting the physical channel into parallel error-free bit pipes.
This can be achieved by treating the physical layer as a standard interference channel (IC) or a standard X-channel (XC).
The IC abstraction gives the network layer two independent bit pipes, each of them providing a DoF of $1/2$ for a sum DoF of $1$.
The XC abstraction can do slightly better by creating four bit pipes of DoF $1/3$ each for a sum DoF of $4/3$.
However, by ``relaxing'' the separation, we can provide to the network layer the same four bit pipes of the XC, but with two additional linear combinations of these bit pipes.
This can improve the performance of the system as soon as caches are available at the receivers.
For example, if each receiver cache can store up to four fifths of a file, then we show below that a sum DoF of $10/3$ can be achieved, compared with $20/9$ for the XC and $5/3$ for the IC for the same memory value.

The physical-layer view of the caching problem adopted in this paper is therefore the Gaussian interference channel together with an X-channel message set, i.e., four messages $V_1, \dots, V_4$, one to be sent from each transmitter to each receiver as shown in \figurename~\ref{fig:physical-layer}.
The physical layer applies real interference alignment \cite{motahari2014alignment} in order to, loosely speaking, allow recovery of the following quantities:
\begin{itemize}
\item Receiver~1 recovers $V_{1}$, $V_{3}$, and $V_{2}+V_{4}$;
\item Receiver~2 recovers $V_{2}$, $V_{4}$, and $V_{1}+V_{3}$.
\end{itemize}

More formally, real interference alignment is a modulation scheme that uses a one-dimensional lattice to create channel inputs $G_i^T$ corresponding to message $V_i$.
Each transmitter then creates the channel inputs
\begin{IEEEeqnarray*}{rCl}
X_1^T &=& h_{22} G_{1}^T + h_{12} G_{2}^T;\\
X_2^T &=& h_{21} G_{3}^T + h_{11} G_{4}^T.
\end{IEEEeqnarray*}
At the channel output, the following signals will be received:
\begin{IEEEeqnarray*}{rCl}
Y_1^T &=& h_{11}h_{22}G_{1}^T + h_{12}h_{21}G_{3}^T + h_{11}h_{12}\left( G_{2}^T+G_{4}^T \right) + \eta_1^T;\\
Y_2^T &=& h_{12}h_{21}G_{2}^T + h_{11}h_{22}G_{4}^T + h_{21}h_{22}\left( G_{1}^T+G_{3}^T \right) + \eta_2^T.
\end{IEEEeqnarray*}

Using the lattice structure of the $G_i^T$'s, user~$1$ can demodulate $G_{1}^T$, $G_{3}^T$, and $G_{2}^T+G_{4}^T$, while user~$2$ can demodulate $G_{2}^T$, $G_{4}^T$, and $G_{1}^T+G_{3}^T$.
Using a linear block code over this modulated channel as proposed in \cite{niesen_whiting2012}, we can ensure that user~$1$ can decode $V_{1}$, $V_{3}$, and $V_{2}\oplus V_{4}$, while user~$2$ can decode $V_{2}$, $V_{4}$, and $V_{1}\oplus V_{3}$.
The addition $\oplus$ here is over some finite field.
For the purposes of this paper, we can assume that this field is $\mathrm{GF}_2$.
If $V_{i}\in[2^{R'T}]$, then the rate
\begin{equation}
\label{eq:rprime}
R' = \frac13 \cdot \frac12\log P + o(\log P),
\end{equation}
corresponding to a DoF of $1/3$, is achievable \cite{motahari2014alignment}.

\subsection{Network-layer view}
\label{sec:setup-network}

The network-layer abstraction replaces the noisy interference channel in \figurename~\ref{fig:setup-gaussian} with the noiseless channel in \figurename~\ref{fig:setup-deterministic}.
The transmitters send four messages $V_1,\ldots,V_4$, each of size $cF$ bits, such that $V_1$ and $V_2$ are sent by transmitter 1 and depend only on $U_1$, and $V_3$ and $V_4$ are sent by transmitter 2 and depend only on $U_2$.
The messages go through the channel and the users receive the following outputs (the symbol ``$\oplus$'' denotes bitwise XOR):
\begin{itemize}
\item
User 1 receives $V_1$, $V_3$, and $V_2\oplus V_4$;
\item
User 2 receives $V_2$, $V_4$, and $V_1\oplus V_3$.
\end{itemize}
The quantity $c$ is called the (normalized) \emph{network load}.
The sum of all the network loads $\rho=4c$ is called the \emph{sum network load}.

A pair $(M,\rho)$ is said to be achievable if a strategy with receiver memory $M$ and sum network load $\rho$ can deliver any requested files to the two users with high probability as the file size $F\to\infty$.
For a cache memory $M$, we call the smallest achievable sum network load $\rho^*(M)$.

\section{Performance Analysis}
\label{sec:results}

We start with a complete characterization of the network layer trade-off in Section~\ref{sec:performance-network}.
We then translate this to the end-to-end system in Section~\ref{sec:gaussian-achievability}, giving an achievability result for the original interference channel.
We also show that in the high-memory regime, our layered architecture is end-to-end optimal.

\subsection{Network-layer performance analysis}
\label{sec:performance-network}

\begin{figure}
\centering
\includegraphics[width=.35\textwidth]{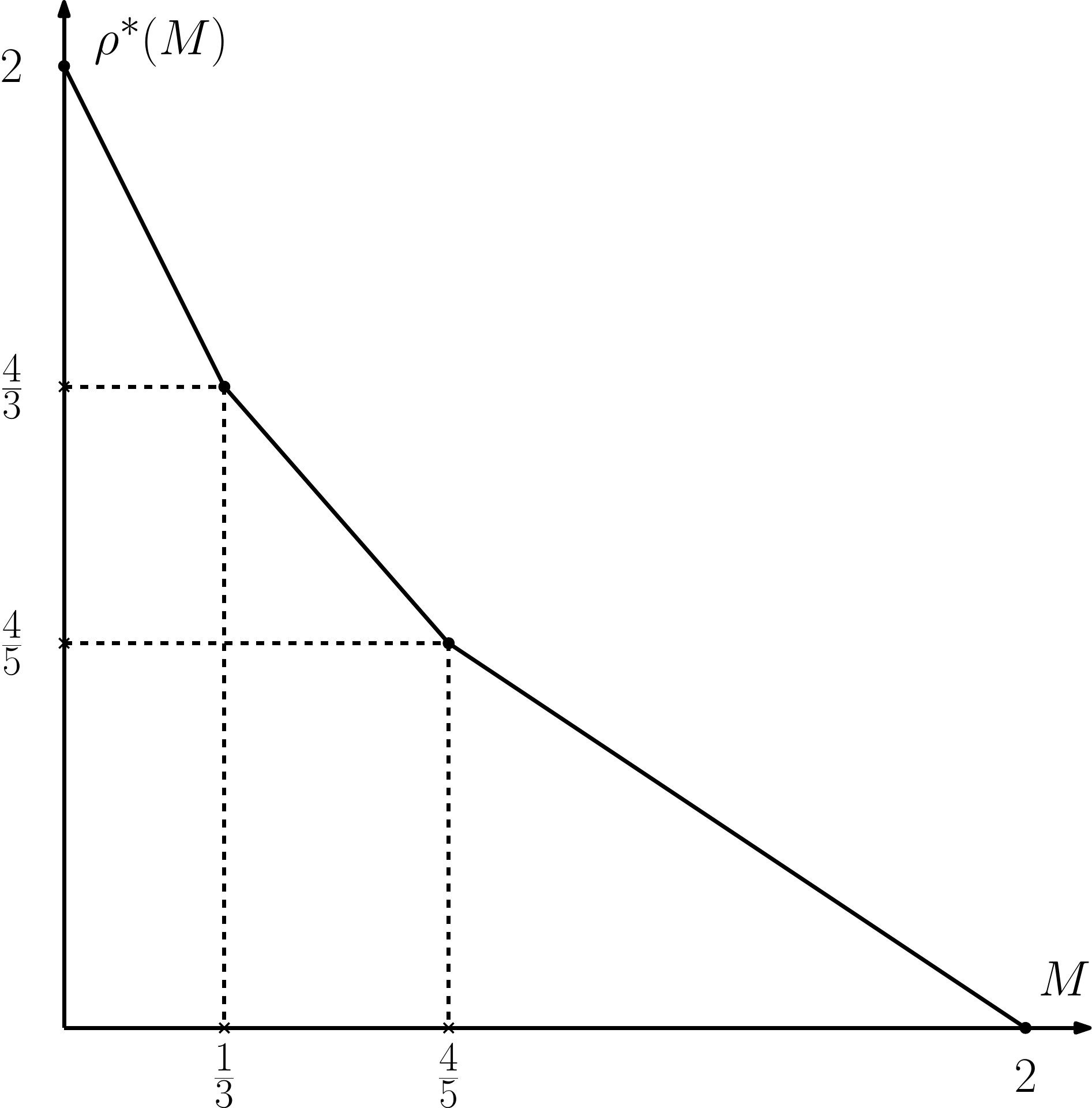}
\caption{Optimal trade-off between sum-network load $\rho$ and receiver memory size $M$ at the network layer.}
\label{fig:deterministic-rate}
\end{figure}

The following theorem, visualized in \figurename~\ref{fig:deterministic-rate}, gives a full characterization of the optimal sum network load $\rho^*(M)$ as a function of receiver cache memory $M$.
\begin{theorem}
\label{thm:deterministic}
At the network layer, the optimal trade-off between $\rho$ and $M$ is:
\begin{equation}
\label{eq:deterministic-optimal}
\rho^*(M) = \max\left\{
2-2M,
\frac{12}{7}-\frac87M,
\frac43-\frac23M,
0
\right\}.
\end{equation}
\end{theorem}

Proving this theorem requires the usual two steps: finding a scheme that achieves the right-hand-side in \eqref{eq:deterministic-optimal}, and proving matching lower bounds.
In the lower bounds, a non-cut-set inequality is needed to show optimality, as cut-set bounds (see \cite[Theorem 15.10.1]{cover}) alone are insufficient.
The details of both the achievability and the lower bounds are given in \appref{sec:deterministic}, but we will here give a brief overview of the ideas involved.

\paragraph*{Overview of the achievability}
The expression in \eqref{eq:deterministic-optimal} is a piece-wise linear function, with the following $(M,\rho)$ corner points: $(0,2)$, $(1/3,4/3)$, $(4/5,4/5)$, and $(2,0)$.
By time and memory sharing, the function $\rho^*(M)$ is a convex function of $M$.
Hence, it suffices to prove that the above corner points are achievable.
We will here give a high-level overview of the scheme for the two most interesting points: $(1/3,4/3)$ and $(4/5,4/5)$.
We will consider only the demand pair $(A,B)$ in the delivery phase, as the others are similar.

\begin{figure}
\centering
\includegraphics[width=.45\textwidth]{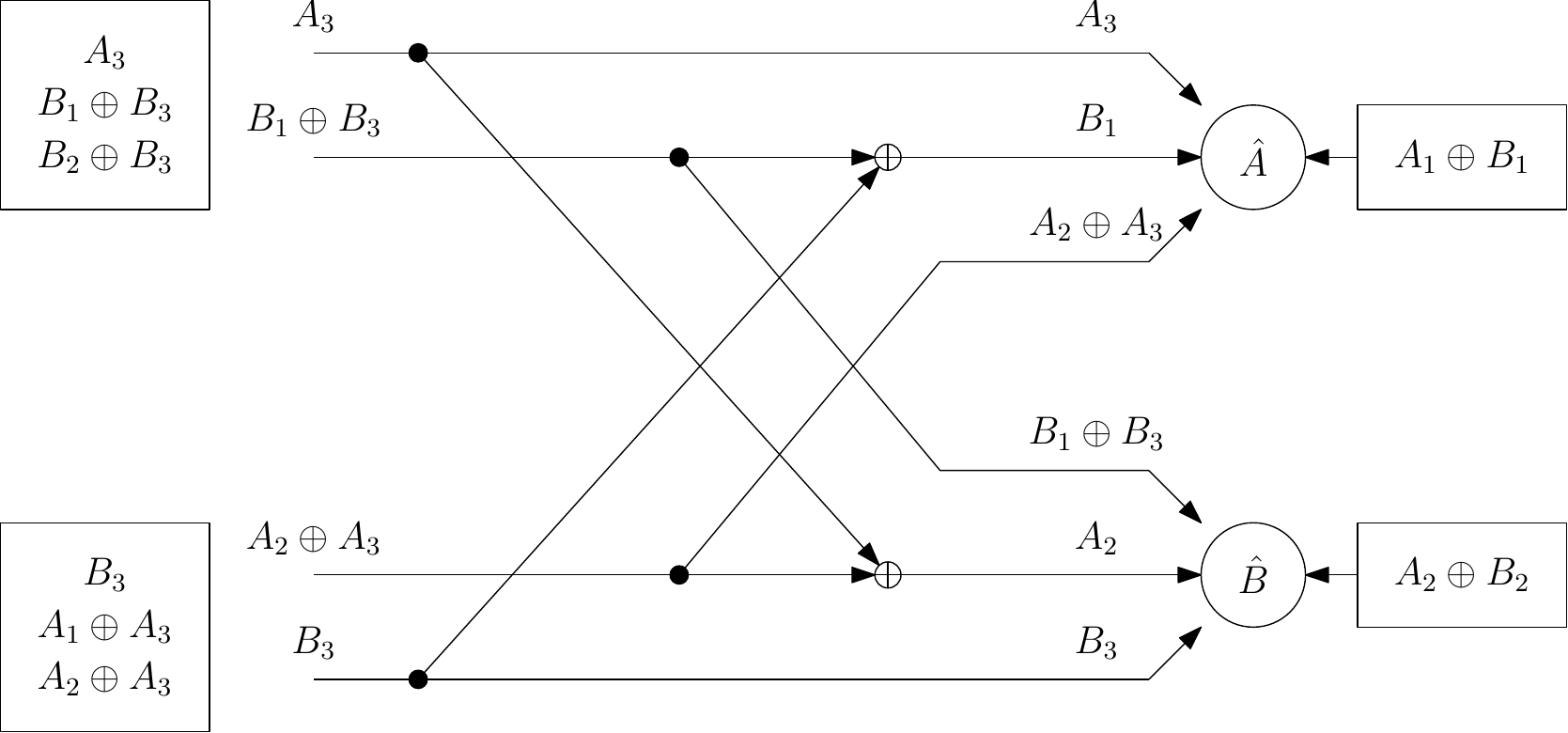}
\caption{Achievable strategy for $M=1/3$ when the demands are $(A,B)$.}
\label{fig:strategy-m13}
\end{figure}

The strategy for point $(M,\rho)=(1/3,4/3)$ is illustrated in \figurename~\ref{fig:strategy-m13}.
If $M=1/3$, then each receiver cache can store the equivalent of one third of a file.
We split each file into three parts: $A=(A_1,A_2,A_3)$ and $B=(B_1,B_2,B_3)$, and store $Z_1=(A_1\oplus B_1)$ and $Z_2=(A_2\oplus B_2)$ at the receivers.
The transmitter caches store $U_1=(A_3,B_1\oplus B_3,B_2\oplus B_3)$ and $U_2=(B_3,A_1\oplus A_3,A_2\oplus A_3)$.
Notice that the contents of $U_1$ and $U_2$ are independent, and that each has a size of $F$ bits.

Suppose now that user~1 requests $A$ and user~2 requests $B$.
Then the transmitters send the messages
\begin{IEEEeqnarray*}{rCl"rCl}
V_1 &=& A_3; &
V_2 &=& B_1\oplus B_3;\\
V_3 &=& A_2\oplus A_3; &
V_4 &=& B_3.
\end{IEEEeqnarray*}
Each $V_i$ has a size of $F/3$ bits, so $c=1/3$ and $\rho=4c=4/3$.

User~1 receives
\begin{IEEEeqnarray*}{rCl}
(V_1,V_2\oplus V_4,V_3)
&=& \left( A_3, (B_1\oplus B_3)\oplus B_3, A_2\oplus A_3 \right)\\
&=& \left( A_3, B_1, A_2\oplus A_3 \right).
\end{IEEEeqnarray*}
The user can recover $A_2$ from $A_2\oplus A_3$ and $A_3$, and decode $A_1$ from $B_1$ and the cache content $A_1\oplus B_1$.
Thus, user~1 has completely recovered file $A$.
User~2 applies a similar approach to decode file $B$.

Note that the transmitters take advantage of the contents of the receiver caches to send a reduced amount of information to the users.
Specifically, they need to communicate $A_3$ to user~1, $B_3$ to user~2, and both $A_2$ and $B_1$ to both users.
Using a similar strategy for the other demands, we see that the point $(M,\rho)=(1/3,4/3)$ is achievable.

\begin{figure}
\centering
\includegraphics[width=.45\textwidth]{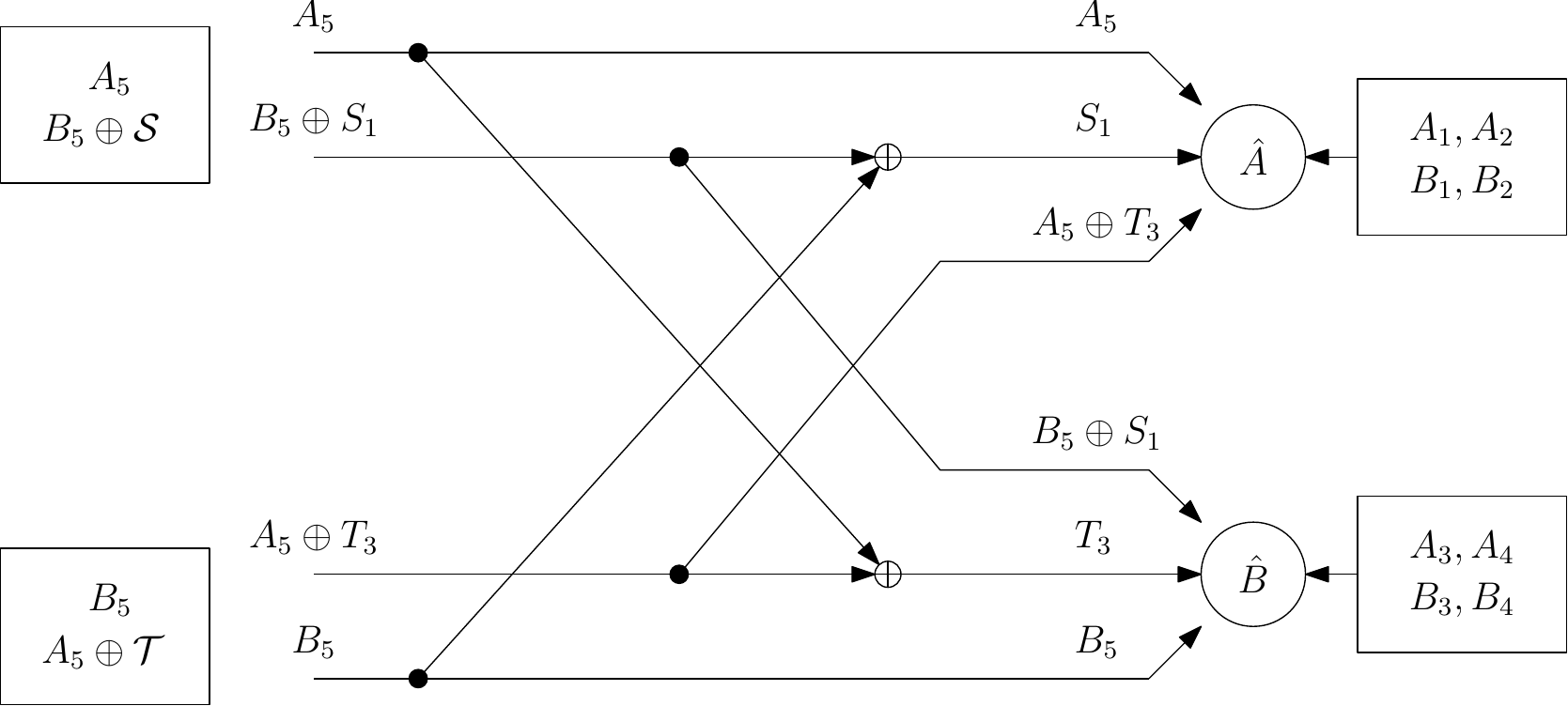}
\caption{Achievable strategy for $M=4/5$ when the demands are $(A,B)$.
We write $\mathcal{S}=\{S_1,S_2,S_3,S_4\}$ and $\mathcal{T}=\{T_1,T_2,T_3,T_4\}$, where $S_1=B_2\oplus A_4$ and $T_3=B_1\oplus A_3$ (the others are not used for demands $(A,B)$).}
\label{fig:strategy-m45}
\end{figure}

Let us now consider point $(M,\rho)=(4/5,4/5)$, whose strategy is visualized in \figurename~\ref{fig:strategy-m45}.
When $M=4/5$, each receiver can store the equivalent of four fifths of a file.
We thus divide each file into five equal parts, $A=(A_1,\ldots,A_5)$ and $B=(B_1,\ldots,B_5)$, and place $Z_1=(A_1,A_2,B_1,B_2)$ and $Z_2=(A_3,A_4,B_3,B_4)$ in the receiver caches.
The transmitter caches, which have a capacity of one file each, store $U_1=(A_5,B_5\oplus \mathcal{S})$ and $U_2=(B_5,A_5\oplus\mathcal{T})$, where $\mathcal{S}$ and $\mathcal{T}$ each consist of four linear combinations
\begin{IEEEeqnarray*}{rClCl}
\mathcal{S} &=& \{S_i\}_{i=1}^4 &=& \{B_2\oplus A_4, A_1\oplus B_3, B_1\oplus B_3, B_2\oplus B_4\},\\
\mathcal{T} &=& \{T_i\}_{i=1}^4 &=& \{A_1\oplus A_3, A_2\oplus A_4, B_1\oplus A_3, A_2\oplus B_4\},
\end{IEEEeqnarray*}
of parts of $A$ and $B$.

Assume the users request files $A$ and $B$, respectively.
The transmitters then send the following messages:
\begin{IEEEeqnarray*}{rClCl}
V_1 &=& A_5;\\
V_2 &=& B_5\oplus S_1 &=& B_5 \oplus \left( B_2\oplus A_4 \right);\\
V_3 &=& A_5\oplus T_3 &=& A_5 \oplus \left( B_1\oplus A_3 \right);\\
V_4 &=& B_5.
\end{IEEEeqnarray*}
Notice that the size of each $V_i$ is $F/5$ bits, so that $c=1/5$ and $\rho=4/5$.

User~1 receives
\begin{IEEEeqnarray*}{rCl}
(V_1,V_2\oplus V_4,V_3) &=& (A_5,B_5\oplus B_5\oplus S_1,A_5\oplus T_3)\\
&=& \left( A_5 , S_1, A_5\oplus T_3 \right).
\end{IEEEeqnarray*}
Recall that user~1's cache already contains $Z_1=(A_1,A_2,B_1,B_2)$.
This gives it the first two parts of $A$, and it receives the fifth part $A_5$ from the channel.
Furthermore, using $A_5$ and $A_5\oplus T_3$ allows it to decode $T_3=(B_1\oplus A_3)$, and using $B_1$ from its cache it can recover $A_3$.
Finally, it can combine $B_2$ from its cache with $S_1=(B_2\oplus A_4)$ to decode the last part, $A_4$.
User~1 has therefore completely recovered file $A$, and in a similar manner user~2 can recover file $B$.
By transmitting similar linear combinations for the other demands, we can show that the point $(M,\rho)=(4/5,4/5)$ is achieved.

\paragraph*{Overview of the converse}
For the outer bounds, we must show that any achievable pair $(M,\rho)$ must satisfy the following inequalities:
\begin{IEEEeqnarray*}{rCl}
\rho &\ge& 2-2M;\\
\rho &\ge& \frac{12}{7}-\frac87M;\\
\rho &\ge& \frac43-\frac23M.
\end{IEEEeqnarray*}
The first and third inequalities can be proved using cut-set bounds.
We will here focus on the second inequality, which requires a non-cut-set argument.
For convenience, we will write it in terms of $c$ instead of $\rho$, and rearrange it as
\[
7c + 2M \ge 3.
\]

Informally, the proof proceeds as follows.
Consider the three outputs of the channel observed by user~1, and suppose they result from user~1 requesting file $A$ (user~2's request is irrelevant).
Call these outputs collectively $\mathbf{Y}^A$.
Then, these outputs combined with cache $Z_1$ should allow user~1 to decode file $A$.
In parallel, consider the four inputs of the channel when user~1 requests file $B$ and user~2 requests file $A$.
Collectively call these inputs $\mathbf{V}^{BA}$.
Combining these inputs with cache $Z_2$ should allow user~2 to also decode file $A$.

So far, user~1 has used $\mathbf{Y}^A$ with its cache to decode $A$, and user~2 has combined $\mathbf{V}^{BA}$ with its cache to also decode $A$.
These decodings have occurred separately from each other.
Let us now combine everything together (i.e., $\mathbf{V}^{BA}$, $\mathbf{Y}^A$, and the two caches $Z_1$ and $Z_2$).
Then, user~1 should decode the remaining file $B$ using $\mathbf{V}^{BA}$ and its cache $Z_1$.

In summary, we have argued that four input messages $\mathbf{V}^{BA}$ and three output messages $\mathbf{Y}^A$, each of which has a size of $cF$ bits, as well as two caches $Z_1$ and $Z_2$, each of which has a size of $MF$ bits, should contain enough information to decode three files ($A$ twice and $B$ once) of size $F$ bits each.
This can be mathematically expressed as
\[
(4+3)cF + 2MF \ge 3F,
\]
thus proving the inequality.
The complete formal proof can be found in \appref{sec:deterministic}.

\subsection{End-to-end performance analysis}
\label{sec:gaussian-achievability}

Applying the physical-layer processing as described in Section~\ref{sec:setup-physical}, the $R'T$ message bits at the physical layer correspond to the $cF$ message bits in the network layer, i.e., $R'T = cF$.
Since the files $A$ and $B$ have a size of $F=RT$ bits, this implies that $R'=cR$.
Therefore, using \eqref{eq:rprime}, the following rate is achievable:
\[
R = \frac{R'}{c} = \frac{1}{3c}\cdot\frac12\log P + o(\log P).
\]
Equivalently, we can achieve the inverse-DoF of $\frac{1}{d(M)}=3c=\frac34\rho$.
To get the largest DoF possible within this strategy, we want to achieve the smallest possible $\rho$, which leads to the following lemma.
\begin{lemma}
\label{lemma:gaussian-dof}
Using the proposed separation architecture, the following end-to-end inverse-DoF is achievable:
\[
1/d(M) = (3/4) \rho^*(M),
\]
where $\rho^*(M)$ is the optimal trade-off between sum-network-load and cache memory at the network layer.
\end{lemma}

A direct combination of Lemma~\ref{lemma:gaussian-dof} and Theorem~\ref{thm:deterministic} leads to the following result.
\begin{theorem}
\label{thm:gaussian-achievability}
The following inverse-DoF is achievable:
\[
\frac{1}{d(M)} = \max\left\{
\frac32-\frac32M,
\frac97-\frac67M,
1-\frac12M,
0
\right\}.
\]
\end{theorem}

The optimality of $\rho^*(M)$ in the network layer implies that our strategy is optimal among all separation-based approaches with the physical-layer processing as proposed here.
In fact, it provides a net improvement over the natural layering strategies: up to a factor of $3/2$ and a factor of $2$ improvement over the X-channel and the interference channel layering schemes, respectively.
In addition, we can show that it is optimal over \emph{all possible strategies} for large receiver cache memory ($M\ge4/5$); see \appref{sec:gaussian} for more details.
Thus, in this regime, the proposed separation into a physical layer and a network layer is without loss of optimality.
We are currently working on extending these results to smaller memories.

\bibliographystyle{IEEEtran}
\bibliography{xchannel,caching}

\ifextended
\appendices

\section{Proof of Theorem~\ref{thm:deterministic}}
\label{sec:deterministic}

In the Appendix, we prove Theorem~\ref{thm:deterministic}, which describes the optimal trade-off between $\rho$ and $M$.
To do that, we first propose an achievable scheme for the setup with trade-off $\rho(M)$, and then give information-theoretic outer bounds on the optimal trade-off $\rho^*(M)$ that match the value of $\rho(M)$.
We formalize this in the following two lemmas.
\begin{lemma}
\label{lemma:deterministic-achievability}
The following trade-off is achievable:
\[
\rho^*(M) \le \max\left\{
2-2M,
\frac{12}{7}-\frac87M,
\frac43-\frac23M,
0
\right\}.
\]
\end{lemma}

\begin{lemma}
\label{lemma:deterministic-converse}
The optimal trade-off $\rho^*(M)$ must satisfy:
\[
\rho^*(M) \ge \max\left\{
2-2M,
\frac{12}{7}-\frac87M,
\frac43-\frac23M,
0
\right\}.
\]
\end{lemma}
Proving these two lemmas is sufficient to prove Theorem~\ref{thm:deterministic}.

\begin{proof}[Proof of Lemma~\ref{lemma:deterministic-achievability}]
To prove Lemma~\ref{lemma:deterministic-achievability}, it suffices to show the achievability of a few $(M,\rho)$ corner points, namely:
\begin{equation*}
(0,2);\quad
(1/3,4/3);\quad
(4/5,4/5);\quad
(2,0).
\end{equation*}
The rest follows using memory-sharing, because $\rho^*(M)$ is a convex function of $M$.
In particular, if two points $(M_1,\rho_1)$ and $(M_2,\rho_2)$ are achievable, then so are:
\[
\left( \lambda M_1+(1-\lambda)M_2,\lambda\rho_1+(1-\lambda)\rho_2 \right),
\]
for all $\lambda\in[0,1]$.

\textbf{Achievability of $(M,\rho)=(0,2)$:}
When $M=0$ the receiver caches are empty.
Split each file into two halves, $A=(A_1,A_2)$ and $B=(B_1,B_2)$, and store them in the transmitter caches as follows: $U_1=(A_1,B_1)$ and $U_2=(A_2,B_2)$.
Note that the $U_i$'s are consistent with the constraint that their size must not exceed that of one file.
Suppose now that the users request files $(W,W')$, where $W,W'\in\{A,B\}$.
Then, the transmitters send the following messages: $V_1=W_1$, $V_2=W_1'$, $V_3=W_2$, and $V_4=W_2'$.
Notice that messages $V_1$ and $V_2$ depend only on $U_1$, while $V_3$ and $V_4$ depend exclusively on $U_2$.

User~1 now receives $(V_1,V_2\oplus V_4,V_3)=(W_1,W_1'\oplus W_2',W_2)$, which allows it to reconstruct $W$.
Similarly, user~2 receives $(V_2,V_1\oplus V_3,V_4)=(W_1',W_1\oplus W_2,W_2')$, which it can use to recover $W'$.
Thus all the demands have been satisfied, and each message sent had a size of exactly half a file, i.e., $cF=F/2$.
Equivalently, $\rho=4c=2$, and the point is achieved.

\textbf{Achievability of $(M,\rho)=(1/3,4/3)$:}
Here, $MF=F/3$, so each receiver cache can store the equivalent of a third of a file.
We start by splitting each file into three parts: $A=(A_1,A_2,A_3)$ and $B=(B_1,B_2,B_3)$.
We then store the following in the receiver caches: $Z_1=(A_1\oplus B_1)$ and $Z_2=(A_2\oplus B_2)$.
This satisfies $M=1/3$.
In the transmitter caches, we place: $U_1=(A_3,B_1\oplus B_3,B_2\oplus B_3)$ and $U_2=(B_3,A_1\oplus A_3,A_2\oplus A_3)$.
Again, the $U_i$'s hold the equivalent of one file, which is consistent with the problem setup.
We show in \tablename~\ref{tbl:delivery-m13} what messages the transmitters should send, for all four possible user demands $(W_1,W_2)\in\{A,B\}^2$, where user~1 requests $W_1$ and user~2 requests $W_2$.
The rows are highlighted in different colors to emphasize which user has access to which information: blue for user~1 and pink for user~2.
Notice that the size of every $V_i$ in the table is exactly one third of a file, which means $cF=F/3$ and $\rho=4c=4/3$ is achievable.

\begin{table}
\centering
\caption{Achievable strategy for $M=1/3$.}
\label{tbl:delivery-m13}
\begin{tabular}{|c|c|c|c|c|c|}
\hline
Cache & \multicolumn{4}{c|}{Content} & User\\
\hline
\rowcolor{user1}
$Z_1$ & \multicolumn{4}{>{\columncolor{user1}}c|}{$A_1\oplus B_1$} & 1\\
\rowcolor{user2}
$Z_2$ & \multicolumn{4}{>{\columncolor{user2}}c|}{$A_2\oplus B_2$} & 2\\
$U_1$ & \multicolumn{4}{c|}{$A_3,B_1\oplus B_3,B_2\oplus B_3$} & N/A\\
$U_2$ & \multicolumn{4}{c|}{$B_3,A_1\oplus A_3,A_2\oplus A_3$} & N/A\\
\hline
\hline
& \multicolumn{4}{c|}{Demands $(W_1,W_2)$} &\\
Message & $(A,A)$ & $(A,B)$ & $(B,A)$ & $(B,B)$ & User\\
\hline
\rowcolor{user1}
$V_1$ & $A_3$ & $A_3$ & $B_2\oplus B_3$ & $B_1\oplus B_3$ & 1\\
\rowcolor{user2}
$V_2$ & $A_3$ & $B_1\oplus B_3$ & $A_3$ & $B_2\oplus B_3$ & 2\\
\rowcolor{user1}
$V_3$ & $A_1\oplus A_3$ & $A_2\oplus A_3$ & $B_3$ & $B_3$ & 1\\
\rowcolor{user2}
$V_4$ & $A_2\oplus A_3$ & $B_3$ & $A_1\oplus A_3$ & $B_3$ & 2\\
\hline
\rowcolor{user1}
$V_2\oplus V_4$ & $A_2$ & $B_1$ & $A_1$ & $B_2$ & 1\\
\rowcolor{user2}
$V_1\oplus V_3$ & $A_1$ & $A_2$ & $B_2$ & $B_1$ & 2\\
\hline
\end{tabular}
\end{table}

\textbf{Achievability of $(M,\rho)=(4/5,4/5)$:}
When $MF=4F/5$, we split each file into five parts: $A=(A_1,\ldots,A_5)$ and $B=(B_1,\ldots,B_5)$.
The placement and delivery are shown in \tablename~\ref{tbl:delivery-m45}, where we have used the following symbols for short:
\begin{IEEEeqnarray*}{rCl'rCl'rCl'rCl}
S_1 &=& B_2\oplus A_4 &
S_2 &=& A_1\oplus B_3 &
S_3 &=& B_1\oplus B_3 &
S_4 &=& B_2\oplus B_4\\
T_1 &=& A_1\oplus A_3 &
T_2 &=& A_2\oplus A_4 &
T_3 &=& B_1\oplus A_3 &
T_4 &=& A_2\oplus B_4
\end{IEEEeqnarray*}
For example, when the demands are $(A,B)$, then user~1 receives $A_5$, $A_5\oplus T_1$, and $S_1$.
This gives it $A_5$ and $S_1=B_2\oplus A_4$, and allows it to decode $T_1=A_1\oplus A_3$.
Using these and the contents of $Z_1=(A_1,A_2,B_1,B_2)$, the user can reconstruct file~$A$.

Notice that the size of each $Z_i$ is $4F/5$ bits, the size of each $U_i$ is $F$ bits, and the size of each $V_i$ is $F/5$ bits, which implies that $\rho=4/5$ is achievable when $M=4/5$.

\begin{table}
\centering
\caption{Achievable strategy for $M=4/5$.}
\label{tbl:delivery-m45}
\begin{tabular}{|c|c|c|c|c|c|}
\hline
Cache & \multicolumn{4}{c|}{Content} & User\\
\hline
\rowcolor{user1}
$Z_1$ & \multicolumn{4}{>{\columncolor{user1}}c|}{$A_1,A_2,B_1,B_2$} & 1\\
\rowcolor{user2}
$Z_2$ & \multicolumn{4}{>{\columncolor{user2}}c|}{$A_3,A_4,B_3,B_4$} & 2\\
$U_1$ & \multicolumn{4}{c|}{$A_5,B_5\oplus S_1,B_5\oplus S_2,B_5\oplus S_3,B_5\oplus S_4$} & N/A\\
$U_2$ & \multicolumn{4}{c|}{$B_5,A_5\oplus T_1,A_5\oplus T_2,A_5\oplus T_3,A_5\oplus T_4$} & N/A\\
\hline
\hline
& \multicolumn{4}{c|}{Demands $(W_1,W_2)$} &\\
Message & $(A,A)$ & $(A,B)$ & $(B,A)$ & $(B,B)$ & User\\
\hline
\rowcolor{user1}
$V_1$ & $A_5$ & $A_5$ & $B_5\oplus S_2$ & $B_5\oplus S_3$ & 1\\
\rowcolor{user2}
$V_2$ & $A_5$ & $B_5\oplus S_1$ & $A_5$ & $B_5\oplus S_4$ & 2\\
\rowcolor{user1}
$V_3$ & $A_5\oplus T_1$ & $A_5\oplus T_3$ & $B_5$ & $B_5$ & 1\\
\rowcolor{user2}
$V_4$ & $A_5\oplus T_2$ & $B_5$ & $A_5\oplus T_4$ & $B_5$ & 2\\
\hline
\rowcolor{user1}
$V_2\oplus V_4$ & $T_2$ & $S_1$ & $T_4$ & $S_4$ & 1\\
\rowcolor{user2}
$V_1\oplus V_3$ & $T_1$ & $T_3$ & $S_2$ & $S_3$ & 2\\
\hline
\end{tabular}
\end{table}

\textbf{Achievability of $(M,\rho)=(2,0)$:}
In this situation, $MF=2F$ allows each user to store both files in its cache.
Therefore, the receiver caches can completely handle the requests, and no messages need ever be sent through the channel.
Hence, $\rho=4c=0$ is achieved.
\end{proof}

\begin{proof}[Proof of Lemma~\ref{lemma:deterministic-converse}]
To prove the lemma, we must show that all the following inequalities hold for any achievable $(M,\rho)$ pair:
\begin{IEEEeqnarray*}{rCl}
\rho &\ge& 2-2M;\\
\rho &\ge& \frac{12}{7}-\frac87M;\\
\rho &\ge& \frac43-\frac23M;\\
\rho &\ge& 0.
\end{IEEEeqnarray*}
The last inequality is trivial.
By using $\rho=4c$, we can rewrite the first three inequalities as:
\begin{IEEEeqnarray}{rCrCl}
\IEEEyesnumber
\label{eq:det-converse-ineqs}
\IEEEyessubnumber
4c &+& 2M &\ge& 2;\label{eq:det-converse-ineq1}\\
\IEEEyessubnumber
7c &+& 2M &\ge& 3;\label{eq:det-converse-ineq2}\\
\IEEEyessubnumber
6c &+&  M &\ge& 2.\label{eq:det-converse-ineq3}
\end{IEEEeqnarray}
Interestingly, inequalities \eqref{eq:det-converse-ineq1} and \eqref{eq:det-converse-ineq3} can be proved using cut-set bounds, but inequality \eqref{eq:det-converse-ineq2} requires non-cut-set-bound arguments.
In proving these inequalities, we introduce the following notation: we use $V_i^{W_1W_2}$ to refer to the input message $V_i$ when user~1 requests file $W_1$ and user~2 requests file $W_2$, where $W_1,W_2\in\{A,B\}$.

To prove \eqref{eq:det-converse-ineq1}, suppose user~1 requests file $A$ and user~2 requests file $B$.
The argument is that all four input messages $(V_1^{AB},\ldots,V_4^{AB})$, each of size $cF$ bits, combined with the two receiver caches $Z_1$ and $Z_2$ of size $MF$ bits each, should contain at least enough information to decode both files, for a total of $2F$ bits.
Thus, $4cF+2MF\ge2F$.
We formalize this using Fano's inequality:
\begin{IEEEeqnarray*}{rCl}
4cF + 2MF
&\ge& H\left( Z_1,Z_2, V_1^{AB},V_2^{AB},V_3^{AB},V_4^{AB} \right)\\
&=& H\left( Z_1,Z_2, V_1^{AB},V_2^{AB},V_3^{AB},V_4^{AB} \middle| A,B \right)\\
&&{} + I\left( A,B;Z_1,Z_2,V_1^{AB},V_2^{AB},V_3^{AB},V_4^{AB} \right)\\
&\ge& H\left( A,B \right)\\
&&{} - H\left( A,B \middle| Z_1,Z_2,V_1^{AB},V_2^{AB},V_3^{AB},V_4^{AB} \right)\\
&\ge& 2F - \epsilon F,
\end{IEEEeqnarray*}
where $\epsilon\to0$ as $F\to\infty$.
Therefore, $4c+2M\ge2$.

For \eqref{eq:det-converse-ineq3}, consider only user~1, requesting both files $A$ and $B$ over two uses of the system.
By using the single cache $Z_1$, and all three outputs of the system $V_1$, $V_3$, and $V_2\oplus V_4$ \emph{twice} (once for each requested file), the user should be able to decode both files.
For simplicity, let $\mathbf{Y}^{W_1W_2}=(V_1^{W_1W_2},V_2^{W_1W_2}\oplus V_4^{W_1W_2},V_3^{W_1W_2})$.
Formally:
\begin{IEEEeqnarray*}{rCl}
6cF + MF
&\ge& H\left( Z_1, \mathbf{Y}^{AB}, \mathbf{Y}^{BA} \right)\\
&\ge& I\left( A,B ; Z_1, \mathbf{Y}^{AB}, \mathbf{Y}^{BA} \right)\\
&=& H\left( A,B \right)
- H\left( A,B \middle| Z_1, \mathbf{Y}^{AB},\mathbf{Y}^{BA} \right)\\
&\ge& 2F - \epsilon F.
\end{IEEEeqnarray*}
Therefore, $6c+M\ge2$.

Finally, inequality \eqref{eq:det-converse-ineq2} requires combining two cut-set bounds.
The first one combines the three output messages of user~1 with its cache to decode file $A$.
In parallel, the second one combines the four input messages with the cache of user~2 to also decode file $A$.
Then, both cut-set bounds are ``merged'' to decode file $B$.
For convenience, let $\mathbf{Y}^{W_1W_2}$ be as defined above, and let $\mathbf{V}^{W_1W_2}=(V_1^{W_1W_2},V_2^{W_1W_2},V_3^{W_1W_2},V_4^{W_1W_2})$.
Formally:
\begin{IEEEeqnarray*}{rCl}
7cF + 2MF
&\ge& H\left( Z_1,\mathbf{Y}^{AB} \right)
+ H\left( Z_2,\mathbf{V}^{BA} \right)\\
&=& H\left( Z_1,\mathbf{Y}^{AB} \middle| A \right)
+ I\left( A ; Z_1,\mathbf{Y}^{AB} \right)\\
&&{} + H\left( Z_2,\mathbf{V}^{BA} \middle| A \right)
+ I\left( A ; Z_2,\mathbf{V}^{BA} \right)\\
&\ge& H\left( Z_1,Z_2,\mathbf{Y}^{AB},\mathbf{V}^{BA} \middle| A \right)\\
&&{} + 2F - 2\epsilon F\\
&\ge& I\left( B ; Z_1,Z_2,\mathbf{Y}^{AB},\mathbf{V}^{BA} \middle| A \right)\\
&&{} + 2F-2\epsilon F\\
&\ge& 3F - 3\epsilon F.
\end{IEEEeqnarray*}
Therefore, $7c+2M\ge3$.
\end{proof}

\section{Converse Results for the End-to-End Problem}
\label{sec:gaussian}

In this Appendix, we provide and prove a lower bound on the optimal trade-off between the inverse DoF and the cache memory.
This lower bound matches the achievable inverse DoF from Theorem~\ref{thm:gaussian-achievability} for $M\ge4/5$.

\begin{theorem}
\label{thm:gaussian-converse}
The optimal inverse DoF obeys the following inequality:
\[
\frac{1}{d^*(M)} \ge \max\left\{ 1 - \frac12M , 0 \right\}.
\]
\end{theorem}

\begin{proof}
Since $1/d^*(M)\ge0$ is trivial, we are left with proving:
\[
\frac{1}{d^*(M)}\ge1-\frac12M.
\]
As with proving \eqref{eq:det-converse-ineq3} in Appendix~\ref{sec:deterministic}, we want to use the fact that a single user should be able to decode both files when using the channel twice.

For convenience, define $X_i^T(W_1,W_2)$ and $Y_j^T(W_1,W_2)$ as the $X_i^T$ and $Y_j^T$ when the requests are $(W_1,W_2)$.
Also define $\mathbf{X}^{W_1W_2}=(X_1^T(W_1,W_2),X_2^T(W_1,W_2))$ and $\mathbf{Y}^{W_1W_2}=Y_1^T(W_1,W_2)$.
\begin{IEEEeqnarray*}{rCl}
2RT
&=& H\left( A,B \right)\\
&=& H\left( A,B \middle| Z_1,\mathbf{Y}^{AB},\mathbf{Y}^{BA} \right)
+ I\left( A,B ; Z_1,\mathbf{Y}^{AB},\mathbf{Y}^{BA} \right)\\
&\le& \epsilon RT
+ I\left( A,B ; Z_1,\mathbf{Y}^{AB},\mathbf{Y}^{BA} \right)\\
&=& \epsilon RT + I\left( A,B ; \mathbf{Y}^{AB},\mathbf{Y}^{BA} \right)\\
&&{} + I\left( A,B ; Z_1 \middle| \mathbf{Y}^{AB},\mathbf{Y}^{BA} \right)\\
&\le& \epsilon RT + I\left( \mathbf{X}^{AB},\mathbf{X}^{BA} ; \mathbf{Y}^{AB},\mathbf{Y}^{BA} \right) + H\left( Z_1 \right)\\
&\le& \epsilon RT + 2 \cdot \frac12\log P \cdot T + MRT.
\end{IEEEeqnarray*}
The last inequality uses the MAC channel bound applied in two instances (demands $(A,B)$ and $(B,A)$).
By taking $T\to\infty$, we get
\[
R\cdot(2-M) \le 2 \cdot \frac12\log P.
\]
Therefore, the optimal DoF must satisfy
\[
d^*(M)\cdot(2-M) \le 2,
\]
which proves the theorem.
\end{proof}

\fi

\end{document}